\pdfoutput=1
\documentclass[a4paper,12pt]{article}

\usepackage{graphicx}
\usepackage{amsmath,amssymb}
\usepackage{fancybox,ascmac}
\usepackage{amsthm}
\usepackage{mathrsfs}
\usepackage{braket}
\usepackage{mathtools}
\usepackage{tikz}
\usetikzlibrary{intersections, calc, arrows.meta}
\usepackage{physics}
\usepackage{hyperref}

\makeatletter
\@addtoreset{equation}{section}
\makeatother

\makeatletter 

\@addtoreset{figure}{section}
\makeatother 

\makeatletter 

\@addtoreset{table}{section}
\makeatother 

\newtheorem{thm}{Theorem}[section]
\newtheorem{dfn}{Definition}[section]
\newtheorem{prop}{Proposition}[section]
\newtheorem{cor}{Corollary}[section]
\newtheorem*{conjecture}{Conjecture}

\title{Relation between spectra of Narain CFTs and properties of associated
boolean functions}
\author{Yuma Furuta\thanks{Reserch Institute for Mathematical Sciences, Kyoto Univercity, 
Kyoto, Japan}
\thanks{Email: yfuruta@kurims.kyoto-u.ac.jp}}
\date{\today}

\begin{document}
\maketitle

\begin{abstract}
    Recently, the construction of Narain CFT from a certain class of quantum error 
    correcting codes has been discovered. 
    In particular, the spectral gap of Narain CFT corresponds to the binary distance
    of the code, not the genuine Hamming distance.
    In this paper, we show that the binary distance is identical to the so-called EPC distance 
    of the boolean function uniquely associated with the quantum code. Therefore, seeking 
    Narain CFT with large spectral gap can be addressed by getting a boolean function 
    with high EPC distance. Furthermore, this problem can be undertaken by finding lower
    Peak-to-Average Power ratio (PAR) with respect to the binary truth table of the boolean function.
    Though this is neither sufficient nor necessary condition for high EPC distance,
    we construct some examples of relatively high EPC distances referring to the constructions 
    for lower PAR. 
    We also see that codes with high distance are related to induced graphs 
    with low independence numbers.

\end{abstract}

\newpage
\tableofcontents

\newpage
\section{Introduction}
There is a very important connection between error correcting codes and 
conformal field theories (CFT for short). 
First, Dolan et al. \cite{dolan1996} discovered the construction of 
a chiral CFT using a classical error correcting code.
Then their construction was extended to a quantum error correcting code and 
Narain CFT by Dymarsky and Shapere
 \cite{dymarsky2020a}. In these cases,
one can relate some quantities of CFTs to 
those of error correcting codes. For instance, in both classical and quantum cases,
self-duality of a code corresponds to 
modular invariance of CFT and the length of code is equal to the central charge of CFT.
In the classical case, the Hamming distance of a code corresponds to the minimum 
value of conformal weight of vertex primary operators, which is known as the spectral 
gap $\Delta$. 

At the level of codes, the Hamming distance plays a very important role in the code theory
because it represents how many errors the code can correct.
For this reason, many researchers have studied the upper and lower bound of the Hamming distance $d$ 
with the ratio of the dimension to the length fixed. 
One of the well-known lower bounds is Gilvert-Varshamov bound $d_{\mathrm{GV}}$ \cite{macwilliams1972}
that guarantees the existence of codes with distance higher than $d_{\mathrm{GV}}$.
Unfortunately, this is a nonconstructive bound, and no systematic way has been found to 
construct codes over $\mathbb{F}_2$ with distance $d\geq d_{\mathrm{GV}}$.
However there are some constructions to obtain asymptotically good codes with $n\to\infty$,
$\frac{d}{n}$ finite where $n$ is the length of the code.
For instance, Muller discovered Reed-Muller code \cite{muller1954} whose length
is $n=2^m$ and distance is $d=2^{m-2}$.
Note that quantum codes that can construct Narain CFT are restricted to real and 
self-dual codes. Since the well-known examples including Reed-Muller code 
are not always real and self-dual, they might not be useful for obtaining 
Narain CFTs with the large spectral gap.

On the other hand, Hartman et. al. \cite{hartman2014}
investigated the condition on $\Delta$ for a Virasoro CFT to have a holographic dual,
and it was extended to Narain CFT cases 
\footnote{In this picture, the holographic dual to the Narain theory is 
a $U(1)$ gravity together with some additional matter field, while 
in the Virasoro case the holographic dual is a pure $U(1)$ gravity.}
by Dymarsky and Shapere \cite{dymarsky2020}.
It requires $\Delta$ to be sufficiently large compared to $c$.
This is analogous to the condition that codes are asymptotically good for large $n$.
So, if $\Delta$ and the Hamming distance of a quantum code are in a sense related to each other,
one can construct various Narain CFTs with large $\Delta$ using constructions 
of good codes.
However, this is not the case because they are slightly different to each other.
Instead, 
one can relate the spectral gap $\Delta$ of a Narain CFT to $d_b$ of corresponding
quantum code. However since $d_b$ is unfamiliar in code theory, no such construction
for good codes with respect to $d_b$ has been known.
Therefore we are motivated to find another representation of $d_b$. 

We find a solution in a boolean function associated with a B-form code, defined in sec \ref{sec3.2}.
Previously, it has been known that the distance of a quantum code is identical to 
APC distance of the associated boolean function. We extend this relation to
$d_b$ and EPC distance of the boolean function. In addition, we find a kind of 
``duality'' between EPC distance and discrete Fourier transformations $\{I,H\}^n$.
The APC distance version of this duality was implied in \cite{danielsen2005}
where it was insisted that codes with high distance tend to be associated with 
boolean functions with low $\mathrm{PAR}$. 
To verify EPC version of this tendency, 
we construct some examples of B-form codes with length $t^2$ and distance 
$2t-2$ for some primes $t$. 

This paper is organized as follows. In section \ref{sec2} we briefly review definitions
of Narain CFT, quantum error correcting codes, and construction of Narain CFT 
using a kind of quantum error correcting codes. Section \ref{sec3} is devoted to interpreting
this construction using terms in graph theory and  boolean functions.
We also introduce definitions of propagation criteria which are used to define APC or 
EPC distance.
Then in section \ref{sec4} we show that the binary distance $d_b$ coincides with
EPC distance, which is our main result. This result suggests that getting higher 
EPC distance leads to obtaining code CFTs with large spectral gaps.
In this section we also study the correlation between high EPC distances and low Peak-to-Average
Power Ratio (PAR). This is done by testing that if codes with high $d_b$ correspond to 
graphs with low independence numbers, since PAR can be represented by the independence number
of the associated graph.

\section{Narain CFT and Quantum Stabilizer Code}\label{sec2}
Recently, some constructions of conformal field theory have been 
studied by \cite{dolan1996,dymarsky2020a} using lattices associated
with even self-dual codes. 
In this section, we review several results of \cite{dymarsky2020a}
where the construction of a Narain CFT associated with a quantum error correcting code
is shown. We also review how the spectral gap is identified with the binary distance of 
a quantum code in favorable cases.

\subsection{Narain CFT}\label{CFT}

Narain CFT was first considered by Narain \cite{narain1986}. It can be obtained by 
compactifying $n$ free bosons $X$ onto an $n$ dimensional torus associated with an 
$n$ dimensional lattice.
Given a lattice $\Gamma\subset\mathbb{R}^n$, one can compactify the field X
to the $n$-dimensional torus $\mathbb{R}^n /\Gamma$:
\begin{equation}
    \vec{X}\sim \vec{X}+2\pi\vec{e},\  \vec{e}\in\Gamma .
\end{equation}
In this configuration, any state of this theory can be labeled by a $2n$ component
vector $(\vec{p}_L,\vec{p}_R)$ where $\vec{p}_L$ and $\vec{p}_R$ have $n$ components respectively.
Using an antisymmetric metric field $\mathcal{B}$,
they are formulated as 
\begin{equation}
    \begin{split}
        \vec{p}_L&=\frac{2\vec{P}+(\mathcal{B}+I)\vec{e}}{2}\\
        \vec{p}_R&=\frac{2\vec{P}+(\mathcal{B}-I)\vec{e}}{2}
    \end{split}
\end{equation}
where $\vec{P}$ is the momentum conjugate to the field $\vec{X}$ and this vector lies in 
$\Gamma^{*}$ ($\Gamma^*$ means the dual lattice of $\Gamma$).
These vectors turn out to satisfy $|\vec{p}_L|^2-|\vec{p}_R|^2\in 2\mathbb{Z}$
so that the set of vectors $(\vec{p}_L,\vec{p}_R)$ for all possible $\vec{P}\in\Gamma^*$
and $\vec{e}\in\Gamma$ form an even self-dual Lorentzian lattice 
$\Lambda\subset\mathbb{R}^{n,n}$. 
By performing a linear transformation
\begin{equation}
    (\vec{p}_L,\vec{p}_R)\mapsto (\alpha,\beta)=\left(\frac{\vec{p}_L+\vec{p}_R}{\sqrt{2}},
    \frac{\vec{p}_L+\vec{p}_R}{\sqrt{2}}\right),
\end{equation}
 the generator matrix $\Upsilon$ and the metric $g$ of the lattice can be expressed as
\begin{equation}
    \Upsilon=
    \begin{pmatrix}
        2\gamma^* & B \\ 0 & \gamma 
    \end{pmatrix}
    ,\ g=
    \begin{pmatrix}
        0 & I\\ I& 0
    \end{pmatrix},
\end{equation}
where $\gamma$ is the generator matrix of $\Gamma$ and $\gamma^*$ is its dual.
Therefore, one finds that any even self-dual Lorentzian 
lattice defines a CFT 
called Narain CFT. To summarize, any Narain CFT corresponds to an even self-dual 
lattice and vise-versa.

\subsection{Error Correcting Codes}
In order to review how to construct Narain CFT from Quantum error correcting code (QECC), 
we give a brief 
introduction to QECC.
QECC is the protocol that enables one
to ``decode" the error state to the original state. 
A qubit of a state $\ket{\psi}$ is regarded as an element in $\mathbb{C}^2$ such that
$(q_0,q_1)$ represents the state $q_0\ket{\uparrow}+q_1\ket{\downarrow}$
where $\ket{\uparrow}$ and $\ket{\downarrow}$ are spin up state
and spin down state respectively. 
Thus all states constitute a Hilbert space $(\mathbb{C}^2)^{\otimes n}$ 
, which is usually denoted $\mathcal{H}$.
In this up-down basis,
one can regard the interactions with the environment as a
linear transformation on $\mathcal{H}$. So,
an error on a state $\ket{\psi}$ is described by a linear operator $\mathcal{E}$
such that the error state is $\mathcal{E}\ket{\psi}$.

Now we will introduce stabilizer code, a kind of QECC that is to be used to construct
Narain CFT.
Let $\mathcal{H}_{\mathcal{C}}$ be a linear subspace of $\mathcal{H}$
that consists of all the states which can be error-corrected.
This is called code subspace.
A stabilizer code associated with $\mathcal{H}_{\mathcal{C}}$ is defined
by a group $\mathscr{S}$
of operators stabilizing any state of $\mathcal{H}_{\mathcal{C}}$, 
called stabilizer group:
\begin{equation}
    \mathscr{S}\coloneqq \left\{ s:\mathcal{H}\to\mathcal{H}\mid s\ket{\psi}=\ket{\psi}
    \mathrm{for}\  ^{\forall}\ket{\psi}\in\mathcal{H}_{\mathcal{C}} 
    \right\}.
\end{equation} 
For simplicity, we assume that any stabilizer $s\in\mathscr{S}$ is obtained 
by taking a tensor product of Pauli matrices up to a constant factor.
In other words, each generator $s_i$ of $\mathscr{S}$ can be expressed
in terms of $n$ components binary vectors
$\alpha,\ \beta$ as
\begin{equation}
    s_i=i^{\alpha\cdot\beta}\epsilon\left(
        X^{\alpha_1}\otimes\cdots\otimes X^{\alpha_n}
    \right)\left(
        Z^{\beta_1}\otimes\cdots\otimes Z^{\beta_n}
    \right),\ \alpha,\beta\in \{0,1\} ^n
\end{equation}
where $\varepsilon$ is a constant phase factor and 
$X, Z$ are Pauli $x$ and $z$ matrices respectively. Note that each $X$ or $Z$
is a linear operator on the corresponding qubit and $\alpha_i,\ \beta_i=1$ 
means Pauli $y$ operator is multiplied on the $i$-th qubit.
Therefore, one can regard each generator $s_i$ as a $2n$ components binary vector
$(\alpha,\beta)$ and commutativity of any two operators requires 
\begin{equation}\label{commutativerelation}
    s_i s_j = s_j s_i \Leftrightarrow \alpha_i \cdot\beta_j - \alpha_j\cdot\beta_i
    \equiv 0 \ \mathrm{mod}2.
\end{equation}
This relation can be simplified by constructing a $(n-k)\times2n$ binary matrix
\begin{equation}
    H = \begin{pmatrix}
        \alpha_1 & \beta_1 \\
        \dots & \dots \\
        \alpha_{n-k} & \beta_{n-k}
    \end{pmatrix}
\end{equation}
where we assume that the number of the generators of $\mathscr{S}$ is $n-k$ for convenience.
Then the relation (\ref{commutativerelation}) becomes $HgH^{\top}=0\ \mathrm{mod}2$
where $g=\begin{pmatrix}
         0&I\\I&0
         \end{pmatrix}$.
This means that the code 
generated by the matrix $H$\footnote{In other words, the code is a subspace of 
$\mathbb{F}_2^{2n}$ which is spanned by $n-k$ rows of $H$.} is even with respect to the 
Lorentzian metric $g$.
So, one can identify a stabilizer code with a corresponding code $\mathcal{C}$ generated 
by rows of $H$.
Moreover, $\mathcal{C}$ can be regarded as an additive code on
Galois field $\mathrm{GF}(4)$\footnote{This is a field extension of GF(2)$=\mathbb{F}_2$ 
by a polynomial $x^2+x+1$, in other words one obtains GF(4) by adding to GF(2) one of
the root $\alpha$ of $x^2+x+1$. So GF(4) can be denoted as $\mathbb{F}_2 (\alpha)$.}.
This can be done using Gray map
\begin{equation}\label{Graymap}
    \begin{split}
        0&\leftrightarrow(0,0),\ 1\leftrightarrow(1,1)\\
        \omega&\leftrightarrow(1,0),\ \bar{\omega}\leftrightarrow(0,1),
    \end{split}
\end{equation}
which is an isomorphism under addition between GF(4) and $(\mathbb{Z}_2)^2$ 
where $\omega$ is one of the cube roots of $1$. Then one gets a codeword $c$ of 
the code on GF(4) by combining the $i$-th components of $\alpha$ and $\beta$ into 
an element of GF(4) via Gray map (\ref{Graymap}), so that $\mathcal{C}$ can be 
regarded as a code on GF(4).

\subsection{Construction of Narain CFT}\label{difinitiondb}
Now we are going to review the result of \cite{dymarsky2020a,
dymarsky2020} where the construction of a Narain CFT from a real self-dual 
quantum stabilizer code $\mathcal{C}$ is introduced. 
This can be done by constructing a Lorentzian even self-dual lattice, associated
with a Narain CFT as described in section \ref{CFT}.
Given a quantum stabilizer code, one can construct an associated
lattice $\Lambda(\mathcal{C})$, defined as 
\begin{equation}\label{lattice}
    \Lambda(\mathcal{C})\coloneqq\left\{
        v\in\mathbb{F}_2^{n}\mid v\equiv ^{\exists}c\ \mathrm{mod}2,\ 
        c\in\mathcal{C}
    \right\}/\sqrt{2}.
\end{equation}
Here the metric of $\Lambda(\mathcal{C})$ is taken to be Lorentzian
and the factor $1/\sqrt{2}$ is necessary for convenience.
Through this construction, some properties are inherited from the code:
if $\mathcal{C}$ is real, $\Lambda(\mathcal{C})$ is even,
and if $\mathcal{C}$ is self-dual, $\Lambda(\mathcal{C})$ is also self-dual.
Thus, a real self-dual quantum stabilizer code $\mathcal{C}$ corresponds to an even self-dual 
lattice with respect to the Lorentzian metric, which can construct a Narain CFT.
In this way, one can construct a Narain CFT from $\mathcal{C}$,
and such a Narain CFT is called code CFT.

Since each element $(\vec{p}_L,\vec{p}_R)$ labels a primary vertex operator 
$
    V_{\vec{p}_L,\vec{p}_R}=:e^{i\vec{p}_L\cdot X_L+i\vec{p}_R\cdot X_R}:,
$
the Euclidean norm $\| (\vec{p}_L,\vec{p}_R)\|^2/2=\frac{p_L^2+p_R^2}{2}$ represents
the conformal dimension of $V_{\vec{p}_L,\vec{p}_R}$.
Note that since the metric of $\Lambda(\mathcal{C})$ is taken to be Lorentzian
the conformal dimensions are not equal to the norm as elements of the lattice
and also not related to the weights of codewords.
However, this difficulty is resolved by defining a new weight of a codeword $c$,
named binary weight as follows.
\begin{equation}\label{binaryweight}
    w_b(c)=w_x(c)+2w_y(c)+w_z(c).
\end{equation}
The difference from the original one is the coefficient on $w_y(c)$.  
Now we define the minimal distance $d_b$ of $\mathcal{C}$ with respect to $w_b$ 
to be the minimum value of $w_b(c)$ among all codewords $c\in\mathcal{C}$. Then $d_b$
is proportional to
the minimal norm of all elements of $\Lambda(\mathcal{C})$, which is twice the smallest
conformal dimension $\Delta$ of corresponding code CFT called spectral gap. 
Strictly, for $n>4$,
$d_b$ is not proportional to $\Delta$ because $\Lambda(\mathcal{C})$
for arbitrary $\mathcal{C}$ always has a vector $\frac{1}{\sqrt{2}}(2,0,\dots ,0)$
whose norm is 2 and is smaller than the spectral gap. This can be partially solved
by applying shift operation \cite{dymarsky2020a}
to elements of $\Lambda(\mathcal{C})$ which makes
it possible to relate $\Delta$ and binary distance $d_b$.
This operation is called ``twist" of a code CFT.
Then, at least for $d_b\leq 4$ one finds the relationship between $d_b$ and corresponding $\Delta$ is
\begin{equation}\label{specralgapanddb}
    \Delta = \frac{d_b}{4}.
\end{equation}
Note that the equation \eqref{specralgapanddb} holds at least for $d_b\leq 4$
\footnote{It may holds for the case $d_b>4$, but it does not always holds.}
since twisted theories also have light states with norms of order 1.
\footnote{We thank Anatoly Dymarsky for reminding us of the fact that 
the binary distance is not always proportional to the spectral gap.}
However, searching higher $d_b$ is a good scientific question, since it is expected that
high binary distances should be related to large spectral gaps, as investigated
in \cite{yahagi2022}.
This question is based on the problem to obtain the holographic dual to the 
Narain CFT, and it imposes the condition that the CFT has a large spectral gap.
Thus, we will investigate how to construct quantum codes with high $d_b$, 
representing $d_b$ by another form. 

\section{Graphs and Boolean Functions}\label{sec3}
After constructing code CFTs from quantum stabilizer codes, we will introduce
in this section the notion of graphs and boolean functions 
associated with real self-dual codes.
They are deeply connected to T-duality of code CFTs and to the corresponding equivalence,
named T-equivalence. Then for boolean functions, one can consider some criteria about 
their periodic characteristic introduced in \cite{danielsen2006}. In section \ref{sec3.2},
we will show that the fixed-aperiodic autocorrelation is related to Hamming distance of 
the quantum code.
\subsection{B-form Codes and Graphs}
For code CFTs, some elements of the code equivalence group named Clifford group can be written 
in terms of T-duality transformations.
It includes arbitrary permutations of codeword components and swapping the $i$-th 
components of $\alpha$ and $\beta$.\footnote{In other words, it is conjugations
of $i$-th component as a GF(4) code, $\omega\leftrightarrow\bar{\omega}$.}
So these transformations are called T-equivalences of a code, and they play 
an important role in discussing B-form codes, defined later.
Note that not all Clifford group elements correspond to T-duality of a code CFT,
because it contains cyclic permutations $1\rightarrow\omega\rightarrow\bar{\omega}
\rightarrow 1$,
which are not T-equivalences.

When a generator matrix of a code is represented in terms of binary vectors $(\alpha,\beta)$
with $2n$ components, one can bring it to the simple form named B-form 
by performing T-equivalence transformations. Explicitly,
\begin{equation}
    G^{\top}=\begin{pmatrix}
        \alpha_1 && \beta_1 \\ \cdots && \cdots \\ \alpha_n && \beta_n
    \end{pmatrix}
    \xrightarrow{\mathrm{T-equivalence}}
    \begin{pmatrix}
        B && I
    \end{pmatrix}
\end{equation}
where $B$ is $n\times n$ matrix and $I$ is $n\times n$ identity matrix.
Moreover, if the code is real self-dual then the matrix $B$ defined above
has zeros on the diagonal and is symmetric, which implies that $B$ corresponds
to the adjacency matrix of a finite graph of size $n$.  
Here the adjacency matrix $B_{ij}$ of a graph of size $n$ is defined as 
$B_{ij}=1$ if and only if vertexes $V_i$ and $V_j$ are connected by an edge
and $0$ otherwise.
Thus, each code CFT can be represented by a corresponding graph, and as in the
discussion later, certain T-dualities can be interpreted as graphical 
transformations named edge local complementation.
Using the adjacency matrix, one can explicitly write down all generators of stabilizer 
group $\mathscr{S}$ as
\begin{equation}
    s_i=\sigma_x^i \prod_{j=1}^{n}(\sigma_z^{j})^{B_{ij}},
\end{equation}
and the unique state $\ket{\psi}$ that is stabilized by all elements in $\mathscr{S}$
\begin{equation}\label{graphstate}
    \ket{\psi}=2^{-n}\sum_{\alpha\in\mathbb{Z}_2^n}
    (-1)^{f(\alpha)}\ket{\alpha_1,\alpha_2,\dots,\alpha_n},\ \alpha=
    (\alpha_1,\alpha_2,\dots,\alpha_n)
\end{equation}
where $f:\mathbb{Z}_2^n\to\mathbb{Z}_2$ is a boolean function associated with the 
adjacency matrix $B_{ij}$ defined as $f(\alpha)=\sum_{i<j} B_{ij}\alpha_i\alpha_j$
for $\alpha\in\mathbb{Z}_2^n$.
In particular, the state $\ket{\psi}$ is called graph state.
One can verify the above equations by referring \cite{hein2006}.

\subsection{Periodic Criterion of Boolean Function}\label{sec3.2}
In this section, we introduce some definitions of propagation criteria of 
boolean functions which are used to construct graph states (\ref{graphstate})
associated with adjacency matrices. Then we also review a part of the results of 
\cite{danielsen2005,danielsen2006}, which we will make use of later.
They claimed some relationships between certain properties of a boolean function
and those of the corresponding quantum self-dual code.

Let $f$ be a function on $\mathbb{Z}_2^n$ valued in $\mathbb{Z}_2$ for a positive integer $n$.
This function is called a boolean function and we will introduce some definitions
of several notions which characterize its periodic or aperiodic properties.
For a given boolean function $f$, one can consider a vector $s$ with the information 
of all values it takes (sometimes this is mentioned as ``truth table").
$s$ is defined as 
\begin{equation}
    s=((-1)^{f(0,0,\dots,0)},(-1)^{f(0,0,\dots,0,1)},\dots,(-1)^{f(1,1,\dots,1)}),
\end{equation}
and usually denoted as $s=(-1)^{f(x)}$. It has $2^n$ components and its $i$-th
component ($0\leq i\leq 2^n -1$) is $(-1)^{f(i_{n-1},i_{n-2},\dots,i_0)}$ where 
$i=\sum_{k=0}^{n-1}i_k 2^k$ is binary representation of $i$.
In terms of graph states (\ref{graphstate}), this can be understood as the 
coefficients of each basis $\ket{i_{n-1},\dots,i_0}$, so in a sense one can regard $s$ 
as a representative of the graph state. Therefore, any action of operators to the
graph state can be identified with their action on $s$. This discussion will appear
later.

Since each component of $s$ is 1 or $-1$, one can consider 2-points discrete Fourier
transformation, which is equivalent to Walsh-Hadamard transformation.
\begin{dfn}(Walsh-Hadamard transformation)
    Let f be a boolean function. Then Walsh-Hadamard transformation (WHT) of its 
    truth table vector $s=(-1)^{f(x)}$ is denoted as $\hat{\chi}_f$ and 
    \begin{equation}\label{WHT}
        \hat{\chi}_f (b)=2^{-\frac{n}{2}}\sum_{x\in\mathbb{Z}_2^n}(-1)^{f(x)}\times(-1)^{b\cdot x}
        =2^{-\frac{n}{2}}\sum_{x\in\mathbb{Z}_2^n}(-1)^{f(x)+b\cdot x}.
    \end{equation}
\end{dfn}
WHT of a truth table $s$ can be easily calculated by multiplying so called 
Hadamard transformation $H_n$ to s.
Here Hadamard transformation is $2^n\times 2^n$ matrix defined as
\begin{equation}
    H_n=H_1\otimes H_{n-1}\ \mathrm{for}\ n\geq 2,\quad H_1=\frac{1}{\sqrt{2}}
    \begin{pmatrix}
        1&1\\1&-1
    \end{pmatrix}
\end{equation} 
where $\otimes$ denotes the tensor product of matrices.
For simplicity, let us discuss only the $n=1$ case and write the $i$-th component 
of $s$ as $(-1)^{f(i)}$ where $i=0,1$. Since the $(i,j)$-element of $H_1$ is 
$(H_1)_{ij}=\frac{1}{\sqrt{2}}(-1)^{ij}$, one can calculate the $i$-th component of 
$H_1 s$ as 
\begin{equation}\label{dFourierandWHT}
    (H_1 s)_i =2^{-\frac{1}{2}}\sum_{j\in\mathbb{Z}_2} (H_1)_{ij}s_j
    =2^{-\frac{1}{2}}\sum_{j\in\mathbb{Z}}(-1)^{ij}(-1)^{f(j)},
\end{equation}
which coincides with the right hand side of (\ref{WHT}) after replacement of 
indices $(i,j)\to(b,x)$.
Because WHT is understood as a discrete Fourier transformation (DFT),
one obtains Wiener-Khintchin's theorem.
\begin{thm}\label{WK}(\textrm{\cite{danielsen2005}})
    Let $f(x)$ be a boolean function and $r(a)$ be its periodic autocorrelation defined 
    as
\begin{equation}
    r(a)=\sum_{x\in\mathbb{Z}_2^n}(-1)^{f(x)+f(x+a)}
\end{equation}
for $a\in\mathbb{Z}_2^n$. Then Wiener-Khintchin's theorem holds 
\begin{equation}\label{WKfomula}
    r(a)
    =\sum_{b\in\mathbb{Z}_2^n}\hat{\chi}_f ^2 (b)(-1)^{b\cdot a}
    .
\end{equation}
\end{thm} 
\begin{proof}
    This can be proved by a straightforward calculation. Starting from right hand side,
    we see that
\begin{equation}
    \begin{split}
        \sum_{b\in\mathbb{Z}_2^n}\hat{\chi}_f ^2 (b)(-1)^{b\cdot a}
        &=
        \sum_{b\in\mathbb{Z}_2^n}\sum_{x,y\in\mathbb{Z}_2^n}
        (-1)^{f(x)+f(y)+b\cdot(x+y)}(-1)^{b\cdot a}\\
        &=
        2^{-n}\sum_{x,y\in\mathbb{Z}_2^n}(-1)^{f(x)+f(y)}\sum_{b\in\mathbb{Z}_2^n}
        (-1)^{b\cdot(x+y+a)}\\
        &=
        \sum_{x,y\in\mathbb{Z}_2^n}(-1)^{f(x)+f(y)}\delta_{x+y+a,0}\\
        &=
        \sum_{x\in\mathbb{Z}_2^n}(-1)^{f(x)+f(x+a)}=r(a).
    \end{split}
\end{equation}
where we used $\sum_{b\in\mathbb{Z}_2^n}(-1)^{b\cdot z}=2^n\delta_{z,0}$ for an arbitrary 
$z\in\mathbb{Z}_2^n$ and $\delta_{z,0}$ is Kronecker's delta, and all the additions
in the exponents of $(-1)$ is taken mod 2.
\end{proof}
Since WHT can be regarded as a spectrum of $\{H\}^n$ transformation where 
$\{H\}^n$ denotes 1-element set of $H\otimes\cdots\otimes H$, Theorem \ref{WK}
means that $\{H\}^n$ spectrum of a boolean function is ``dual'' of its periodic 
autocorrelation in terms of Fourier transformation. Then when one extends 
this relation to larger sets, for instance $\{I,H,N\}^n$ defined later,
the autocorrelation function obtained by the extension of (\ref{WKfomula})
will be the fixed-aperiodic autocorrelation function of the boolean function.
Moreover, one can define a distance with respect to the fixed-aperiodic autocorrelation
which measures its aperiodic autocorrelation property in a sense, which will turn out
to be equal to the distance of the code associated with the boolean function.
Therefore, one conjectures that the distance of the code and the $\{I,H,N\}^n$ spectrum
of the boolean function are related. From now, we will explain this relationship 
due to \cite{danielsen2005,danielsen2006}.

First, we remark some notations used in \cite{danielsen2005}.
For two binary vectors $x,y\in\mathbb{Z}_2^n$,
\begin{equation}
    x\succeq y\Leftrightarrow x_i\geq y_i \ \mathrm{for}\  ^{\forall}i \in \{1,\dots,n\}.
\end{equation}
Then for a binary vector $a$, its conjugate $\bar{a}$ is 
$\bar{a}_i=a_i+1\ \mathrm{mod}\ 2$, and
a vector space $V_{a}$ is
\begin{equation}
    V_{a} \coloneqq \{ x\in\mathbb{Z}_2^n \mid x\preceq a\}.
\end{equation}
So, for $k\preceq a$, $k+V_{a}$ defines a coset space of $V_{a}$ where
$k+V_{a}=\{k+x\mid x\preceq a\}$. 

Next, we introduce the definition of the (fixed-) aperiodic autocorrelation function
of a boolean function and its distance called APC distance.
\begin{dfn}
    Let $f$ be a boolean function. Its aperiodic autocorrelation function $s(a,k)$
    $(a,k\in\mathbb{Z}_2^n)$
    is defined as 
    \begin{equation}
        s(a,k)=\sum_{x\in k+V_{\bar{a}}}(-1)^{f(x)+f(x+a)},\quad k\preceq a,
    \end{equation}
    and its fixed-aperiodic autocorrelation function $s(a,\mu,k)$ $(\mu\in\mathbb{Z}_2^n)$ is
    \begin{equation}
        s(a,\mu,k)=\sum_{x\in k+V_{\bar{\mu}}}(-1)^{f(x)+f(x+a)},\quad a,k\preceq \mu.
    \end{equation}
\end{dfn}

Similarly to Theorem \ref{WK}, $s(a,\mu,k)$ is ``dual'' of some spectra with respect
to certain transformations including WHT. To explain this, we introduce another transformation
than WHT named nega-Hadamard transformation $N$ (NHT). It is defined as 
\begin{equation}
    N=\frac{1}{\sqrt{2}}\begin{pmatrix}
        1&i\\1&-i
    \end{pmatrix}.
\end{equation}
Then the transformation set that is ``dual'' to $s(a,\mu,k)$ turns out to be 
$\{I,H,N\}^n$, which consists of all transformations of form
\begin{equation}
    \prod_{i\in\mathbf{R}_I}I \prod_{j\in\mathbf{R}_H}H_j \prod_{k\in\mathbf{R}_N}
    N_k.
\end{equation}
Here $(\mathbf{R}_I,\mathbf{R}_H,\mathbf{R}_N)$ is the partition of $\{1,2,\dots,n\}$
and $H_j$ denotes $I\otimes I\otimes\cdots\otimes I\otimes H\otimes I\otimes\cdots\otimes I$
with $H$ acting only on the $j$-th tensor component and $I$ on the others. So $\{I,H,N\}^n$
has $3^n$ components. By the same argument as WHT \eqref{dFourierandWHT},
one finds that the spectrum of a $\{I,H,N\}^n$ transformation
\footnote{An $\{I,H,N\}^n$ transformation means an element of $\{I,H,N\}^n$.} 
is of the following form \cite{riera2006}
\begin{equation}
    P_{k,c,r,\mu}=2^{-\frac{n-wt(\mu)}{2}}\sum_{x\in r+V_{\bar{\mu}}}
    (-1)^{f(x)+k\cdot x}(i)^{wt_c (x)}
\end{equation}
where $wt_c (x)$ denotes the weight restricted on the support
\footnote{Support of $c$ is the set of indices defined as $\{i\in\{1,\dots,n\}\mid c_i=1\}$.
} of $c$ and $k,c\in V_{\bar{\mu}}$, $r\in V_{\mu}$. 
Then the next proposition will reveal ``duality'' between $\{I,H,N\}^n$ 
spectrum and a certain aperiodic autocorrelation function.
\begin{prop}
For a boolean function $f(x)$ and $\mu,k,c\in\mathbb{Z}_2^n$ where $a,k\preceq\mu$
and $c\preceq \bar{\mu}$, let $\mathbf{R}_N$ be the indices on which 
nega-Hadamard transformation is acted. Then set $c_i =1 \Leftrightarrow i\in\mathbf{R}_N$.
Then one finds the following equation 
\begin{equation}\label{eq3.22}
    \sum_{x\in k+V_{\bar{\mu}}} (-1)^{f(x)+f(x+a)+\sum_{i}\chi_{\mathbf{R}_N}(i)a_i(x_i+1)}
    =(i)^{-wt_c (a)}\sum_{u\in V_{\bar{\mu}}}\| P_{u,c,k,\mu}\|^2(-1)^{u\cdot a}.
\end{equation}
where $\chi_{R_N}(i)$ is the characteristic function of $\mathbf{R}_N$, 
\begin{equation}
    \chi_{R_N}(i)=
    \begin{cases}
        1 & (i\in \mathbf{R}_N)\\
        0 & (i\notin \mathbf{R}_N)
    \end{cases}.
\end{equation}
\end{prop}
\begin{proof}
    This can be proved in the same way as Theorem \ref{WK}.
    \begin{equation}
        \begin{split}
        2^{-(n-wt(\mu))}&\sum_{u\in V_{\bar{\mu}}}\sum_{x,y\in k+V_{\bar{\mu}}}
        (-1)^{f(x)+f(y)+u\cdot(x+y+a)}(i)^{wt_c (x)-wt_c (y)-wt_c (a)}\\
        &=\sum_{x,y\in k+V_{\bar{\mu}}}(-1)^{f(x)+f(y)}(i)^{wt_c (x)-wt_c (y)-wt_c (a)}
        \delta^{\bar{\mu}}_{x+y+a,0}
        \end{split}
    \end{equation}
    where $\delta^{\bar{\mu}}_{x+y+a,0}$ means $\delta^{\bar{\mu}}_{x+y+a,0}=1$ if and
    ony if $(x+y+a)|_{\bar{\mu}}=0$. Then when $y|_{\bar{\mu}}=(x+a)|_{\bar{\mu}}$,
    \begin{equation}
        \begin{split}
             wt_c(x)-wt_c(y)-wt_c (a)&=wt_c (x)-wt_c (x+a)-wt_c (a)\\
            &=wt_c (x)-(wt_c (x)+wt_c (a)-\\
            &\sum_{i}\chi_{\mathbf{R}_N}(i)a_i x_i)-wt_c (a)\\
            &=2\sum_{i}\chi_{\mathbf{R}_N}(i)a_i x_i -2wt_c (a)\\
            &=2\sum_{i}\chi_{\mathbf{R}_N}(i)a_i(x_i-1),
        \end{split}
    \end{equation}
    which leads one to \eqref{eq3.22} because $a_i(x_i-1)$ and $a_i(x_i+1)$ make no 
    difference on the exponent of (-1).
\end{proof}

In a cryptographic sense,
Danielsen et.al \cite{danielsen2006} used $s(a,\mu,k)$ to define aperiodic propagation criteria (APC).
They were motivated in a block cipher scenario, which we will not argue for more detail,
and they defined APC distance to measure to which degree the boolean function satisfies APC.
\begin{dfn}\label{defofAPC}
    For positive integers $(l,q)$, a boolean function $f$ satisfies APC($l$) of order $q$
    if $s(a,\mu,k)=0$ for any $a,k,\mu\in\mathbb{Z}_2^n$ such that $a,k\preceq\mu$
    and $1\leq w(a)\leq l,\ 0\leq w(\mu \& \bar{a})\leq q$ where 
    $(x\& y)_i=1\Leftrightarrow x_i=y_i=1$. Here $W(x)$ denotes the weight of $x$,
    and if for a integer $d>0$ $f$ satisfies APC($l$) of order $q$ for all $(l,q)$
    such that $d>l+q$, $f$ has APC distance $d$.
\end{dfn}
Although APC distance was defined in a cryptographic context, it can be 
related to the code theory via the boolean function. This is described in 
the following theorem.
\begin{thm}\label{distanceAPCQECC}(\cite{danielsen2006})
    Let $f(x)=\sum_{i<j}B_{ij}x_i x_j$ for $B_{ij}\in\{0,1\}\ (i\neq j),\ B_{ii}=0$ 
    be a boolean function with APC distance $d$. Then 
    QECC $\mathcal{C}$ associated with $f(x)$ constructed by the matrix 
    $B_{ij}$ is real self-dual and has distance $d$.
\end{thm}
\begin{proof}
    See appendix \ref{appendixB}.
\end{proof}
Wiener-Khintchin theorem with respect to 
$\{I,H,N\}^n$ transformation implies that power spectrum
$\| P_{k,c,r,\mu}\|$ is related to APC, as discussed in \cite{danielsen2005}. Specifically, they 
defined peak-to-average ratio with respect to $\{I,H,N\}^n$ transformation
$\mathrm{PAR}_{I,H,N}$ \cite{danielsen2005a} as for $s=(-1)^{f(x)}$
\begin{equation}
    \mathrm{PAR}_{I,H,N}(s)=\max_{U\in\{I,H,N\}^n,k\in\mathbb{Z}_2^n}\{\|(Us)_{k}\|^2\}
\end{equation}
and investigated the correspondence to values of APC distance.
They concluded that these two values are related in the way
\begin{equation*}
    \mathrm{low\ PAR}_{I,H,N}\leftrightarrow \mathrm{high\ APC\ distance}.
\end{equation*}
However, this is just a tendency and some pairs of them are counterexamples.

\section{Main Results}\label{sec4}
As in Theorem \ref{distanceAPCQECC}, the distance of a QECC $\mathcal{C}$ can be identified with
APC distance of the associated boolean function and as reviewed in sec \ref{difinitiondb}
the binary distance $d_b$ of $\mathcal{C}$ is proportional to the spectral gap of 
its associated Narain CFT
for a limited condition. However because the definition of $d_b$ is subtlety
different from that of $d$, APC distance may not be related to spectral gap $\Delta$.
In this section, we will introduce another periodic criterion, extended periodic
criterion (EPC), and its distance (EPC distance), and show that EPC distance 
can be associated with $\Delta$. 
\subsection{Extended Periodic Criteria and Distance}\label{sec4.1}
EPC was first defined by Preneel \cite{preneel1991a} and investigated by Caret 
\cite{CARLET199932} to extend the configuration in which aperiodic propagation
criteria were considered.
\begin{dfn}
    For a boolean function $f(x)$, the fixed-extended autocorrelation function
    $v(a,k,\mu)$
    for $a,k,\mu\in\mathbb{Z}_2^n$, $k\preceq\mu$ is defined by
    \begin{equation}
        v(a,k,\mu)=\sum_{x\in k+V_{\bar{\mu}}}(-1)^{f(x)+f(x+a)}.
    \end{equation}
    $f$ is set to have EPC distance $d$ if $f$ satisfies EPC($l$) of order $q$
    for $l+q<d$,
    where EPC($l$) of order $q$ means that $v(a,k,\mu)=0$ for all
    $k\preceq\mu$ such that $1\leq w(a)\leq l$ and $0\leq w(\mu)\leq q$.
\end{dfn}
From the above definition, one finds that EPC can be seen as an extension of APC 
in that taking $a\preceq\mu$ reduces $v(a,k,\mu)$ to $s(a,k,\mu)$. This observation 
leads us to 
\begin{equation*}
    \mathrm{APC\ distance}\leq\mathrm{EPC\ distance}.
\end{equation*}
Extending the identity between APC distance and code distance, we have the following theorem.
\begin{thm}\label{thm4.1}
    For a boolean function $f(x)$ and the binary distance $d_b$ of a quantum code  
    associated with $f(x)$, EPC distance and $d_b$
    are related as
    \begin{equation}
        \mathrm{EPC\ distance\ of\ }f=d_b.
    \end{equation}
\end{thm}
Similarly to the reformulation of APC distance in Appendix \ref{appendixB},
we use another form of the definition of EPC distance found by Preneel \cite{preneel2003}:
\begin{prop}\label{eq4.3}
    EPC distance of a boolean function is equal to the least weight of the 
    $2n$ component binary vector $(a,b)$ such that 
    \begin{equation}
        \sum_{x\in\mathbb{Z}_2^n}(-1)^{f(x)+f(x+a)+b\cdot x} \neq 0.
    \end{equation}
\end{prop}
Eq.\eqref{eq4.3} is similar to the one of the case of APC distance, the only difference
being the way of counting weight of a vector $(a,b)$.
In the APC distance case, the weight is the number of indices such that $a_i$ 
or $b_i$ is 1, which is consistent with the weight of a codeword in QECC. 
On the other hand, in the EPC distance case,
the indices such that $a_i =b_i =1$ are doubly-counted. The way of counting
is consistent with the binary weight \eqref{binaryweight} because $w_y (c)$ is 
also doubly-counted whose representation in $(\mathbb{Z}_2)^2$ via \eqref{Graymap}
is $(1,1)$. Thus, considering eq. \eqref{specralgapanddb} and 
applying the discussion in appendix \ref{appendixB} to the 
case of EPC distance and binary weight, one finds Theorem \ref{thm4.1}.

Therefore, when investigating the spectral gap of a code CFT, one can use 
EPC distance of the corresponding boolean function associated with its B-form.
Moreover, to find Narain CFTs with large spectral gaps, instead, it is partially resolved 
by making EPC distance higher. This topic is discussed in the next 
subsection.

As can be seen in definitions of APC distance and EPC distance, these two 
distances usually take different values.
For instance, consider the boolean function $f(x)$ made from a clique graph
whose edges are all possible edges that connect arbitrary two vertexes.
In this case, the adjacency matrix is $B_{ij}=1$ for $i\neq j$ and the associated
boolean function $f(x)$ is $f(x)=\sum_{i<j}x_i x_j$. Then the APC distance and 
EPC distance are two and four respectively, and distance $d$ and binary distance $d_b$ 
can be easily seen to be two and four.

\subsection{EPC distance and PAR}
Generally, the construction of CFT with large spectral gap or QECC with high 
distance is so difficult a problem that no systematic construction has been found.
However as in the above discussion, codes with high distance might be related to 
boolean functions with low $\mathrm{PAR}_{I,H,N}$. For instance, Danielsen \cite{danielsen2005} studied 
nested clique graphs and found some nested clique graphs of length $\leq 30$ can 
construct codes with high distance, some of which have optimal distance.
They are motivated in \cite{parker2002}, which considered how to make $\mathrm{PAR}_{I,H,N}$
lower.
In this subsection, similarly to APC distance,
we will show the connection between EPC distance and PAR with respect
to certain discrete Fourier transformations. Then we also extend examples
of nested clique graphs in \cite{danielsen2005} to larger lengths with high distances.

First, extending Theorem \ref{WK} to the transformations in $\{I,H\}^n$, 
one finds 
\begin{thm}
    For a boolean function $f(x)$, the fixed-extended autocorrelation function
    satisfies 
    \begin{equation}\label{eq4.4}
        v(a,\mu,k)=\sum_{u\in V_{\bar{\mu}}}|P_{u,k,\mu}|^2 (-1)^{u\cdot a}.
    \end{equation}
    where $a\preceq\bar{\mu}$, $k\preceq\mu$ and 
    \begin{equation}
        P_{u,k,\mu}=\sum_{x\in k+V_{\bar{\mu}}}(-1)^{f(x)+u\cdot x}
    \end{equation}
    is the spectrum of a $\{I,H\}^n$ transformation.
\end{thm}
This implies that partially the fixed-extended autocorrelation function is 
Fourier transformation of the spectrum of $\{I,H\}^n$ transformation. Therefore 
one can write the spectrum as a Fourier transformation of $v(a,\mu,k)$,
\begin{equation}
    |P_{u,k,\mu}|^2 = \sum_{a\preceq\bar{\mu}}v(a,\mu,k)(-1)^{a\cdot u}.
\end{equation}
This equation sets the upper bound of $|P_{u,k,\mu}|^2$.
Let $d$ be the EPC distance of $f(x)$. Considering 
the definition of $v(a,\mu,k)$ and using the triangle inequality 
out of the region where $1\leq w(a)\leq l,\ 0\leq w(\mu)\leq m$ and 
$l+m<d$, one obtains
\begin{cor}
    Each spectrum of $\mathrm{PAR}_{I,H,N}$ is bounded from above.
    \begin{equation}\label{upperbound}
        |P_{u,k,\mu}|^2\leq 2^{n-w(\mu)}  \left\{ 
            \sum_{i=d-w(\mu)}^{n-w(\mu)}\left(
                \begin{array}{c}
                    n-w(\mu)\\
                    i
                \end{array}
            \right)+1
        \right\}.
    \end{equation}
\end{cor}
As known in \cite{danielsen2005a}, it is self-evident that 
$\mathrm{PAR}_{I,H,N}\leq 2^n$, and this corollary reveals stricter upper bound.
So it is a convincing argument that higher EPC distance is related to lower 
$\mathrm{PAR}_{I,H,N}$. 
Then we will construct some examples of the so-called nested clique graphs
with low $\mathrm{PAR}_{I,H,N}$ in the next subsection.
\subsection{Examples of nested clique graphs}
According to Parker and Tellambura \cite{parker2002}, one can construct
boolean functions with lower $\mathrm{PAR}_{I,H,N}$ as follows.
\begin{thm}(\cite{parker2002})
    Let $p(x)=p(x_0,\dots,x_{n-1})$ be a boolean function defined as 
    \begin{equation}
        p(x)=\sum_{j=0}^{L-2}\theta_j(\mathbf{x}_j)\gamma_j(\mathbf{x}_j)
        +\sum_{j=0}^{L-1}g_j(\mathbf{x}_j)
    \end{equation}
    where $\theta_j,\gamma_j$ are any permutations $\mathbb{Z}_2^t\to\mathbb{Z}_2^t$
    and 
    \begin{equation}
        \mathbf{x}_j=\{x_{\pi(tj)},x_{\pi(tj+1)},\dots,x_{\pi(t(j+1)-1)} \},
    \end{equation}
    $\pi$ is a permutation of $\mathbb{Z}_n$.
    Then $s=2^{-\frac{1}{2}}(-1)^{p(x)}$ satisfies $\mathrm{PAR}(s)\leq 2^t$
    where $\mathrm{PAR}$ denotes the peak-to-average power ratio with respect to 
    linear unimodular unitary transform (LUUT)\footnote{
        Now we will not introduce the definition of LUUT, but this does not 
        affect later arguments.
    }.
\end{thm}
This theorem can bound $\mathrm{PAR}_{I,H,N}$ from above because $\{I,H,N\}^n$
transformations are included in the class of LUUTs. Danielsen \cite{danielsen2005}
investigated whether this construction can give lower $\mathrm{PAR}_{I,H,N}$.
At the level of graphs, this construction includes a class of graphs 
known as nested clique graphs. Clique graphs are defined at the last of section 
\ref{sec4.1}, and a nested clique graph has a nested structure of clique graphs. 
For instance, if one nests 3-clique graph to 2-clique graph, the result is 
a graph shown in fig \ref{K2K3}. 
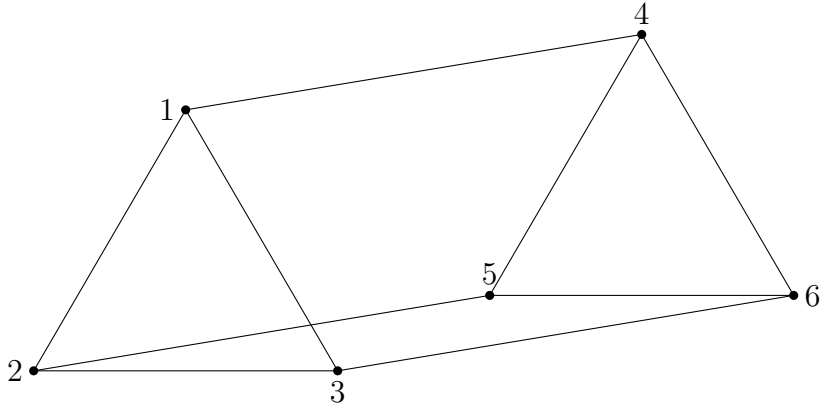
\begin{figure}[tbp]
    \centering
\begin{tikzpicture}
    \coordinate [label=left:1] (1) at (2,3.46); 
    \coordinate [label=left:2] (2) at (0,0); 
    \coordinate [label=below:3] (3) at (4,0); 
    \foreach \P in {1,2,3} \fill[black] (\P) circle (0.06);  
    \draw (1)--(2)--(3)--(1); 
    \coordinate [label=above:4] (5) at (8,4.46); 
    \coordinate [label=above:5] (6) at (6,1); 
    \coordinate [label=right:6] (7) at (10,1); 
    \foreach \Q in {5,6,7} \fill[black] (\Q) circle (0.06);
    \draw (5)--(6)--(7)--(5);
    \draw (1)--(5);
    \draw (2)--(6);
    \draw (3)--(7);
\end{tikzpicture}
\caption{An example of $[K_2[K_3]]$ graph.}\label{K2K3}
\end{figure}
Here the nested clique graph in this example is denoted $[K_2[K_3]]$
where $K_n$ means a clique graph with $n$ vertexes.
In the following, we will show several examples of $[K_t[K_t]]$ where $t\in\mathbb{Z}_{>0}$.
is a prime.
It consists of $t$ blocks of clique graph where every block is a $K_t$ graph and 
each block is denoted as $K_t^{(i)}$ for $i=1,2,\dots,t$.
So, for an arbitrary pair of two blocks ($K_t^{(i)}$,$K_t^{(j)}$),  
the vertexes of $K_t^{(i)}$,$K_t^{(j)}$ are connected in the following way.
If one regards the sets of all edges which connect $K_t^{(i)}$ and $K_t^{(j)}$
as a map from the vertex set of $K_t^{(i)}$ to that of $K_t^{(j)}$,
it is bijective.
Therefore choosing the bijection $\sigma^{i,j}:K_t^{(i)}\to K_t^{(j)}$ for all $i\neq j$,
one gets a nested clique graph. Here the bijection can be regarded as an element of 
the permutation group $S_n$.
Then relabeling vertexes, $\sigma^{i,i+1}$ for $i=1,\dots,t-1$ can be fixed to be 
the identity map. Therefore, when constructing a nested clique graph, 
one has to choose $\sigma^{i,j}$ for $j\geq i+2$ and the number of them is 
$\frac{1}{2}(t-1)(t-2)$.

We introduce a way of construction of a nested clique graph $[K_t[K_t]]$ for 
a prime number $t\geq 5$.
First, label all $\sigma^{i,j}$ for $j\geq i+2$ as $\sigma_1,\sigma_2,\dots,\sigma_{s}$
where $s=\frac{1}{2}(t-1)(t-2)$.
The order of labeling can be taken arbitrarily.
Then let $k,l,m$ be integers such that $1\leq k\leq s$, $k=lt+m$, and $1\leq m\leq t$.
Then $\sigma_k$ is set as 
\begin{equation}\label{perms}
    \sigma_k = \left(
        \begin{array}{llll}
            1&2&\cdots&t\\
            m&m+l&\cdots&m+(t-1)l
        \end{array}
    \right).
\end{equation}
For example, $\sigma_1$ is the identity permutation $(1\ 2\ \cdots\ t)$.
The condition that $t$ is prime is needed for $(m\ m+l\ \cdots\ m+(t-1))$ to be 
a permutation of $t$.
Then the adjacency matrix is constructed as 
\begin{equation}
    B=
    \begin{pmatrix}
        1_t-I & I & \sigma & \cdots & \sigma_2 & \sigma_1\\
        I & 1_t-I & I & \sigma &\cdots & \sigma_3\\
        \sigma^{\top} & I & 1_t-I & I &\cdots & \vdots\\
        \vdots & \vdots & I & \ddots & \ddots & \vdots\\
        \sigma_2^{\top} & \vdots & \vdots & \ddots & 1_t -I & I \\
        \sigma_1^{\top} & \sigma_3^{\top} & \cdots & \cdots & I & 1_t-I
    \end{pmatrix}.
\end{equation}
Here $1_t$ is $t\times t$ matrix with all entries 1 and $\sigma_k$ denotes 
$t\times t$ matrix with $(\sigma_k)_{i\sigma_k(i)}=1$ and 0 for otherwise.
For $t=3$, one has to choose one permutation $\sigma_1$ not to be the identity permutation,
and one gets $9\times 9$ adjacency matrix 
\begin{equation}
    B=
    \begin{pmatrix}
        0 & 1 & 1 & 1 & 0 & 0 & 0 & 1 & 0\\
        1 & 0 & 1 & 0 & 1 & 0 & 0 & 0 & 1\\
        1 & 1 & 0 & 0 & 0 & 1 & 1 & 0 & 0\\
        1 & 0 & 0 & 0 & 1 & 1 & 1 & 0 & 0\\
        0 & 1 & 0 & 1 & 0 & 1 & 0 & 1 & 0\\
        0 & 0 & 1 & 1 & 1 & 0 & 0 & 0 & 1\\
        0 & 0 & 1 & 1 & 0 & 0 & 0 & 1 & 1\\
        1 & 0 & 0 & 0 & 1 & 0 & 1 & 0 & 1\\
        0 & 1 & 0 & 0 & 0 & 1 & 1 & 1 & 0
    \end{pmatrix}
\end{equation}
for which the associated QECC generated by
$(I\mid B)$ has binary distance $d_b=4$. 

By this construction, we verified that for $t=3,5,7,11$ the associated QECC has 
at least $d_b=2t-2$. In fact, we constructed QECCs with $n=9,25,49,121$ and their 
binary distances are calculated to be $d_b=4,8,12,21$ respectively. 
For large $t$, this value is not so 
high because $d_b/n\to 0$ when 
$n(=t^2)\to\infty$. 
However, $d_b /n$ decreases as slowly as $1/\sqrt{n}$.
This means that these codes have binary distances as high as the square root of $n$.
Therefore, if this holds in general, we can construct relatively good codes 
compared to already known good codes associated with graphs.
Then we have the following conjecture, which is remained to be proved or disproved.
\begin{conjecture}
Let $t$ be an odd prime number.
$B$-form code $\mathcal{C}$ constructed by the nested clique graph with respect to 
permutations \eqref{perms} has length $t^2$ and binary distance at least $2t-2$.
\end{conjecture}

\subsection{$\mathrm{PAR}_{I,H}$ and independence number}
As we have seen, the binary distance of a code is related to $\mathrm{PAR}_{I,H}$ of the
associated boolean function. So we are motivated to find quadratic boolean
functions with lower $\mathrm{PAR}_{I,H}$. On the other hand, Riera and Parker \cite{riera2006b}
showed that $\mathrm{PAR}_{I,H}$ is related to the size of the
maximum independent set (i.e. independence number), which is explained in the following. 
In this subsection, we examine the relationship between the binary distance and the
independence number of the graph by computing these values for 
some examples.

First, we explain the definition of the maximum independent set of a graph and 
the proposition that relates $\mathrm{PAR}_{I,H}$ and independence number. 
\begin{dfn}
    Let $G=(V,E)$ is a graph of size $n$. Then a subset $W$ of $V$ is called an 
    independent set of $G$ if and only if no two vertexes of $W$ are connected.
    If $|W|$ denotes 
    the size of $W$, the maximum value of $|W|$ for independent sets of $G$ is 
    called independence number $\alpha(G)$ of $G$. The independent set with this maximum value 
    is the maximum independent set. 
\end{dfn}
For instance, for the graph in Fig.\ref{K2K3}, the vertex set $\{1,5\}$ is an independent
set of $[K_2[K_3]]$ of size two and one finds this is the maximum independent set. 
Therefore $[K_2[K_3]]$ has independence number two. Note that the maximum independent set 
is not unique. The vertex set $\{2,4\}$ in the above example has size two,
so this is also the maximum independent set. As in the following proposition, 
Riera and Parker showed that the independence number deals with $\mathrm{PAR}_{I,H}$.
\begin{prop}(\cite{riera2006b})
    Let $f(x)=\sum_{i<j}B_{ij}x_i x_j$ be a boolean function associated with a 
    adjacency matrix $B=(B_{ij})$ of a graph $G$. Then $\mathrm{PAR}_{I,H}$ of 
    $(-1)^{f(x)}$ and $\alpha(G)$ is related as 
    \begin{equation}
        \mathrm{PAR}_{I,H}=2^{\alpha(G)}.
    \end{equation}
\end{prop}
With this proposition, one can evaluate binary distance by measuring $\alpha(G)$
instead of $\mathrm{PAR}_{I,H}$. Now we will observe the tendency 
\begin{equation*}
    \mathrm{low\ }\alpha(G)\leftrightarrow \mathrm{high\ EPC\ distance}.
\end{equation*}
through several examples of graphs, 
comparing the binary distance (or Hamming distance) and independence number.
\subsubsection*{Example 1 : Nested clique graphs}
The first example is nested clique graphs $[K_t[K_t]]$ which we have constructed in the
previous section. For these graphs, the independence number can be easily computed 
to be $t$ as follows.
The independence number $\alpha([K_t[K_t]])$ is obviously bounded from above $\alpha([K_t[K_t]])\leq t$ 
because for each complete graph $K_t^{(i)}$ no two vertexes can be disconnected.
Then choosing one vertex from each $K_t^{(i)}$ adequately, one obtains an independent set of 
size $t$, which saturates the bound. Since we have already established the conjecture
about the binary distances of nested clique graphs, we should check whether the independence 
number $t$ is low. To check this, we compare the independence number $\alpha([K_t[K_t]])=t$ to
those of random graphs.

Let $\alpha(G(n,m))$ be an independence number of a random sparse graph $G$ with $m$ edges and of size $n$.
Then it is well-known that a non-constructive argument gives 
\begin{equation}\label{eq4.14}
    \alpha(G(n,m))\sim \frac{2\log d_V}{d_V}n
\end{equation}
where $d_V$ is the vertex degree of $G$. This means $\alpha(G(n,m))$ is nearly equal to 
$\frac{2\log d_V}{d_V}n$ with high possibility.
Note that this approximation can be applied only for sparse graphs, which is the case 
for nested clique graphs.
In this case, one has $n=t^2$ and $d_V=2t-2$. Assigning these values to the right hand side of 
eq.\eqref{eq4.14}, we see 
\begin{equation}
    \alpha(G(n,m))\sim \frac{t^2\log (2t-2)}{t-1}
\end{equation}
with high possibility. This is apparently higher than $\alpha([K_t[K_t]])=t$ and 
the ratio of those independence numbers is 
\begin{equation}
    \frac{\alpha([K_t[K_t]])}{\alpha(G(n,m))}\sim o (\frac{1}{\log t})
\end{equation}
for large $t$. So we take these graphs as examples that have 
relatively high binary distances and low independence numbers.
\subsubsection*{Example 2 : Circulant graphs}
Next we consider circulant graphs,
which were investigated by Grassl and Harada \cite{GRASSL2017399}.
A circulant graph is a graph whose adjacency matrix is circulant, and
a circulant matrix is defined to be of the following form:
\begin{equation}\label{eq4.17}
    \begin{pmatrix}
        r_0 & r_1 & \cdots & r_{n-2} & r_{n-1}\\
        r_{n-1} & r_0 & \cdots & r_{n-3} & r_{n-2}\\
        \vdots & \vdots & \ddots & \vdots & \vdots\\
        r_2 & r_3 & \cdots & r_0 & r_1\\
        r_1 & r_2 & \cdots & r_{n-1} & r_0
    \end{pmatrix}
\end{equation} 
where $r_0=0$ and each $r_i$ for $i=1,\dots,n-1$ is zero or one.
We list in the table \ref{tablegrassl} the codes they discovered and 
the independence numbers of their corresponding graphs.  
Here $n$ and $d_H$ are respectively
the length and Hamming distance of the code. We also list the independence number 
$\alpha(G)$ with adjacency matrix of form \eqref{eq4.17}.
In order to verify that these codes have low $\alpha(G)$, we compare these independence
numbers to those of many random graphs as in the table.

\begin{table}[h]
    \centering
    \begin{tabular}{c|ccc}
        \hline
        name & ($n,d_H$) &  $\alpha(G)$ & ($\alpha(G)$, $\#$ of random graphs)\\
        \hline\hline
        $C_{56}$ & (56,15)  & 7 &(7,180), (8,7773), (9,2011), (10,36)\\
        $C_{57}$ & (57,15)  & 11 &(10,211), (11,6129), (12,3445), (13,208), (14,7)\\
        $C_{58}$ & (58,16)  & 8 &(8,1754), (9,7453), (10,786), (11,7)\\
        $C_{63}$ & (63,16)  & 8 &
        
            (10,4), (11,3017), (12,6188), (13,769), (14,21), (15,1)
        \\
        $C_{67}$ & (67,17)  & 6 &(7,4), (8,5792), (9,4103), (10,100), (11,1)\\
        $C_{70}$ & (70,18)  & 4 &(5,13), (6,8418), (7,1560), (8,9)\\
        $C_{71}$ & (71,18)  & 9 &(9,3530), (10,6133), (11,330), (12,7)\\
        $C_{79}$ & (79,19)  & 9 &(7,436), (8,8816), (9,744), (10,4)\\
        $C_{83}$ & (83,20)  & 5 &(7,2233), (8,7495), (9,271), (10,1)\\
        $C_{87}$ & (87,20)  & 8 &(8,246), (9,8502), (10,1235), (11,17)\\
        $C_{89}$ & (89,21)  & 8 &(9,1976), (10,7561), (11,458), (12,4), (13,1)\\
        $C_{95}$ & (95,20)  & 7 &(8,5207), (9,4704), (10,89)\\
        \hline
    \end{tabular}
    \caption{Codes that Grassl and Harada \cite{GRASSL2017399} constructed 
    compared to random graphs.}\label{tablegrassl}
\end{table}
For each code listed in table \ref{tablegrassl}, we enumerate 10,000 random 
graphs with the same length and vertex degree as the corresponding code.
Then computing the independence numbers of all these graphs, we list 
in the right column pairs of an independence number and the number of 
random graphs with the independence number.
The result implies that many of the codes that they constructed have low 
independence number compared to random graphs. In particular, 
for $C_{63},\ C_{67},\ C_{70},\ C_{83},\ C_{89},\ C_{95}$, their independence numbers 
are lower than the lowest number among corresponding 10,000 random graphs. 

\subsubsection*{Example 3 : Circulant matrices}
In addition to the result above, Harada \cite{harada2018} investigated 
several graphs whose adjacency matrices are given by pairs of circulant matrices. 
He constructed three adjacency matrices combining two circulant matrices $A$ and $B$
as 
\begin{equation}\label{eq4.18}
    \begin{pmatrix}
        A & B \\ B^{\top} & A
    \end{pmatrix},
\end{equation}
and found corresponding codes have higher distances than the previously known upper bound
as in table \ref{tableharada}.

\begin{table}[h]
    \centering
    \begin{tabular}{c|ccc}
        \hline
        name & ($n,d_H$) &  $\alpha(G)$ & ($\alpha(G)$, $\#$ of random graphs)\\
        \hline\hline
        $C_{66}$ & (66,17)  & 5 & (7,25), (8,7254), (9,2661), (10,60)\\
        $C_{78}$ & (78,19)  & 7 & (7,21), (8,7670), (9,2174), (10,35)\\
        $C_{94}$ & (94,21)  & 8 & (9,466), (10,8422), (11,1106), (12,6)\\
        \hline
    \end{tabular}
    \caption{Three codes that Harada \cite{harada2018} constructed compared to 
    random graphs. }\label{tableharada}
\end{table}
Similarly to Example 2, one finds that these codes have lower independence numbers than 
10,000 random graphs.
This example together with Example 2 suggests that
low independence number may be regarded as an index for high distance.

\subsubsection*{Example 4 : Complete graphs}
Since the correspondence between high distance and low $\mathrm{PAR}$ is not rigorous,
there are some counterexamples with low distances and low independence numbers.
One of them is complete graphs, which is denoted as $K_t$ in this paper.
As in the last of section \ref{sec4.1}, the binary distance and Hamming distance 
are known to be four and two, which is low because these values are constant for any length $n$.
On the other hand, the independence number of a complete graph is one,
as all pairs of vertexes are connected. This independence number is quite lower 
than random graphs listed in table \ref{tablegrassl}, \ref{tableharada}.

\section{Conclusion}
In this paper, we discussed how the binary distance of a quantum code
is related to its boolean function. 
On the other hand, the binary distance of QECC is proportional to the spectral gap of 
the Narain CFT in a favorable condition and is expected to be related to 
the spectral gap in general.
Therefore, our result can be used when seeking Narain CFTs with large spectral gaps,
replacing the problem by that of getting higher EPC distances.
However in general, making EPC distance higher is very difficult, so we planned to obtain lower PAR instead of 
higher EPC distances.
According to previous works by Danielsen \cite{danielsen2005}, 
high APC distance seems to have a high chance of implying low PAR and vise versa.
Extending this relation to EPC distance, 
we formulated the upper bound of $\mathrm{PAR}_{I,H}$ with its bound depending on
the EPC distance. 

Moreover, consulting with the work of Parker and Tellambura \cite{parker2002}, we constructed some 
nested clique graphs of length $t^2$ and confirmed that their binary distances satisfy 
$d_b\geq 2t-2$ for $t=3,5,7,11$. To prove or disprove that this holds for all $t$, and to 
make the upper bound \eqref{upperbound} stricter are left to future works.
We also tested the correlation between the high distance and the low independence number
comparing the independence number of graphs with high distance to those of 
a large number of random graphs.
As a result, we found that many graphs follow the expectation, though some examples do not.
Therefore, we propose the low independence number as a possible index for high distance. 

Let us conclude this paper with some comments on the future direction.
As we construct some quantum codes with relatively high binary distances, 
it should be tested how large spectral gaps the Narain CFTs associated with these codes 
have. In addition, it is a remaining problem to construct the holographic duals to the code 
theories. This paper will shed light on this problem by proposing the construction of 
a quantum code with a high chance of a large binary distance.
The physical interpretation of the construction of Narain CFTs with quantum codes
is also to be given. Recently this has been studied by \cite{buican2021a,dymarsky2021b}.
Finally, it is a future subject to extend the correspondence between Narain CFTs and 
quantum stabilizer codes. In the construction of \cite{dymarsky2020a}, 
only GF(4) codes regarded as codes over $\mathbb{F}_2$ are used.
Some extensions are given in \cite{dymarsky2021c,yahagi2022} where the coefficient field 
is extended to some finite fields $\mathbb{F}_p$. Using some constructions 
of good codes over $\mathbb{F}_p$ for $p>2$, one may address the problem to 
obtain a larger spectral gap of Narain CFT.

\section*{Acknowledgements}
We thank Toshiya Kawai and Anatoly Dymarsky for discussions and comments on our paper.
\appendix
\section{The relation between APC distance and code's distance}\label{appendixB}
In this appendix, we explain how to identify APC distance with distance of QECC,
following \cite{danielsen2005,danielsen2006}.

As mentioned in sec. \ref{sec3.2}, operations on stabilized state 
\begin{equation}
    \mathcal{E}:\ket{\psi}\mapsto\mathcal{E}\ket{\psi},\ 
    \ket{\psi}=\sum_{\alpha\in\mathbb{Z}_2^n}(-1)^{f(\alpha)}\ket{\alpha_1,\dots,\alpha_n}
\end{equation}
are identified with transformations on the associated boolean function $f(x)$.
In particular, error operators on $\ket{\psi}$ is generated by Pauli operators
$\{I,\sigma_x ,\sigma_y ,\sigma_z\}$ and each action on the boolean function can
explicitly be written in terms of certain vectors $a,b$. 
First we consider $\sigma_x$ operator on $\ket{\psi}$. Let $\mathcal{E}$ be
an operator associated with a vector $a=(a_i)\in\mathbb{Z}_2^n$ such that it makes 
$\sigma_x$ act on the $i$-th qubit if and only if $a_i=1$ and $I$ for otherwise.
\begin{equation}
    \mathcal{E}=\bigotimes_{a_i=1}\sigma_x^{(i)}\bigotimes_{a_j=0}I^{(j)}.
\end{equation}
Then the action on $f(x)$ is easily verified to be 
\begin{equation}
    f(x)\mapsto f(x+a).
\end{equation}
This is because $\sigma_x$ is bit-flip operator $\sigma_x \ket{0}=\ket{1}$, so 
performing $\sigma_x$ on the $i$-th qubit means that the $i$-th argument of $f(x)$ is
flipped as $0\leftrightarrow 1$. Since in $\mathbb{Z}_2^n$, the flip $0\leftrightarrow 1$
is equivalent to the addition $+1$, $f(x)$ changes to $f(x+a)$.
Then consider $\sigma_z$ operators associated with a vector $b=(b_i)\in\mathbb{Z}_2^n$.
\begin{equation}
    \mathcal{E}=\bigotimes_{b_i=1}\sigma_z^{(i)}\bigotimes_{b_j=0}I^{(j)}.
\end{equation}
The operator $\sigma_z$ has eigenvectors $\ket{0},\ket{1}$ whose eigenvalues are 
1 and -1 respectively. So performing $\sigma_z$ on the $i$-th qubit state $\ket{\alpha_i}\ (\alpha=0,1)$
is equivalent to adding $\alpha_i$ to $f(x)$.
Therefore the action of $\mathcal{E}$ written in terms of $f(x)$ is of form 
\begin{equation}
    f(x)\mapsto f(x)+b\cdot x.
\end{equation}
Combining these two observations, one obtains the action of $\sigma_y=i\sigma_x\sigma_z$.
Thus, the error operator $\mathcal{E}$ is 
\begin{equation}\label{error}
    \mathcal{E}=\bigotimes_{a_i=1,b_i=0}\sigma_x^{(i)}\bigotimes_{a_j=b_j=1}\sigma_y^{(j)}
    \bigotimes_{a_k=0,b_k=1}\sigma_z^{(k)}\bigotimes_{a_l=b_l=0}I^{(l)},
\end{equation}
and $f(x)$ changes under the action of $\mathcal{E}$ as 
\begin{equation}
    f(x)\mapsto f(x+a)+b\cdot x.
\end{equation}

The distance of QECC is defined as the minimum weight of error operators $\mathcal{E}$
\eqref{error} such that QECC can detect the error where the weight is the number of
indices on which $\sigma_{x,y,z}$ is performed. Note that the weight is equal to
the weight of $(a,b)$ regarded as a codeword of GF(4) code. Then in order for
QECC to detect an error defined by $(a,b)$, it is necessary and sufficient that 
the original state $\ket{\psi}$ and the error state $\mathcal{E}\ket{\psi}$ are orthogonal.
Since $\{\ket{\alpha_1\dots\alpha_n}\}$ is an orthogonal basis, the orthogonality
of $\ket{\psi}$ and $\mathcal{E}\ket{\psi}$ is equivalent to the condition
\begin{equation}
    s\cdot s'=0\quad s=(-1)^{f(x)},\ s'=(-1)^{f(x+a)+b\cdot x}.
\end{equation}  
Therefore, the distance of QECC is identical to the minimum value of the
weight of the vector $(a,b)$
such that 
\begin{equation}
    \sum_{x\in\mathbb{Z}_2^n}(-1)^{f(x)+f(x+a)+b\cdot x}\neq 0.
\end{equation}
Then this is the same form as the definition of APC distance reformulated in \cite{danielsen2006},
which is equivalent to Definition \ref{defofAPC}.
So one can identify APC distance of a boolean function with the distance of the corresponding 
QECC.

\newpage
\bibliographystyle{jalpha}
\bibliography{Furuta202203}

\end{document}